\documentclass[11pt,english]{article}

\pdfoutput=1 

\usepackage{scalefnt} 

\usepackage[utf8]{inputenc}
\usepackage[T1]{fontenc}
\usepackage[english]{babel}
\usepackage[dvipsnames]{xcolor} 
\usepackage{amsmath} 
\usepackage{amssymb} 
\usepackage{amsthm}  
\usepackage{hyperref}  
\usepackage{tikz}      
\usepackage{tcolorbox} 
\usepackage{dcolumn}   
\usepackage{graphicx, wrapfig, subcaption, setspace, booktabs, float, epsfig} 
\usepackage{multirow}  
\usepackage[all]{xy}

\usepackage{enumitem}
\usepackage{tikz}

\theoremstyle{plain}
\newtheorem{theorem}{Theorem}[section]

\newtheorem{proposition}{Proposition}
\newtheorem{lemma}{Lemma}[section]

\newtheorem{corollary}{Corollary}
\newtheorem{example}{Example}
\newtheorem*{existence}{Existence Problem}
\newtheorem*{obstruction}{Obstruction Problem}
\newtheorem*{maximality}{Maximality Problem}
\newtheorem*{universality}{Universality Problem}
\theoremstyle{definition}
\newtheorem{remark}{Remark}

\hypersetup{
colorlinks=true,
linkcolor=blue,
citecolor=blue,
filecolor=blue,
urlcolor=black,
pdfauthor={Luiz Felipe Andrade Campos},
pdftitle={Reducible Yang-Mills Theories}}

\begin{document}

\title{On Maximal, Universal and Complete Extensions of Yang-Mills-Type Theories}

\author{Yuri Ximenes Martins\footnote{yurixm@ufmg.br (corresponding author)}, Luiz Felipe Andrade Campos\footnote{luizfelipeac@ufmg.br},  Rodney Josu\'e Biezuner\footnote{rodneyjb@ufmg.br}\\ \\ \textit{Departamento de Matem\'atica, ICEx, Universidade Federal de Minas Gerais,}  \\  \textit{Av. Ant\^onio Carlos 6627, Pampulha, CP 702, CEP 31270-901, Belo Horizonte, MG, Brazil}}


\maketitle
\begin{abstract}
In this paper we continue the program on the classification of extensions of the Standard Model of Particle Physics started in \cite{category_extension}. We propose four complementary questions to be considered when trying to classify any class of extensions of a fixed Yang-Mills-type theory $S^G$: existence problem, obstruction problem, maximality problem and universality problem. We prove that all these problems admits a purely categorical characterization internal to the category of extensions of $S^G$. Using this we show   that maximality and universality are dense properties, meaning that if they are not satisfied in a class $\mathcal{E}(S^G;\hat{G})$, then they are in their ``one-point compactification'' $\mathcal{E}(S^G;\hat{G})\cup \hat{S}$ by a specific trivial extension $\hat{S}$. We prove that, by means of assuming the Axiom of Choice, one can get another maximality theorem, now independent of the trivial extension $\hat{S}$. We consider the class of almost coherent extensions, i.e, complete, injective and of pullback-type, and we show that for it the existence and obstruction problems have a complete solution. Using again the Axiom of Choice, we prove that this class of extensions satisfies the hypothesis of the second maximality theorem.
\end{abstract}

\section{Introduction}
\quad\;\, While Standard Model of Elementary Particle Physics (SM) is probably the most accurate and tested physical theory \cite{review_particle_physics}, there are indications that it must be viewed as an effective theory of a more fundamental extended theory \cite{beyond_SM_1,beyond_SM_2}. In the last decades a cascade of models were proposed \cite{string,higher_dimension_1,higher_dimension_2,QG_1,QG_2,SM_extension_1,SM_extension_2,SM_extension_3,next_minimal_SM,next_minimal_SM_2,next_minimal_SM_3}. In \cite{category_extension} the authors started a program aiming to classify all the possible extensions of SM, starting with the Yang-Mills part. We showed that most of the existing proposals of extensions for YM theories, such as \cite{YM_deformation_1, YM_deformation_2, YM_deformation_CS, tensor_gauge_1, tensor_gauge_2, tensor_gauge_3, SW_map, forgacs1980, Harnad, PhysRevD.99.115026}, are particular examples of a simple and unifying notion of extension. We also proved that the Higgs mechanism and the YM theories in an spinorial background can also be regarded as examples of this unifying notion of extension, defined as follows.

Let $S^G:\operatorname{Conn}(P;\mathfrak{g})\rightarrow \mathbb{R}$ the action functional of a YM theory with gauge group $G$ and instanton sector $P$. An \textit{extension} of $S^G$ consists of an extension $\hat{G}$ of $G$, which itself induces an extension $\hat{P}$ of $\hat{G}$, and another action functional $\hat{S}^{\hat{G}}:\widehat{\operatorname{Conn}}(\hat{P};\hat{\mathfrak{g}})\rightarrow \mathbb{R}$, playing the role of the \textit{extended action funcional}, such that, at least when restricted to a subspace $C^1(\hat{P};\hat{\mathfrak{g}})$, it can be written as 
\begin{equation}\label{extension_intro}
  \hat{S}^{\hat{G}}[\alpha]=S^G[D(\alpha)]+C[\alpha],  
\end{equation}
i.e, as a sum of the original YM theory (evaluated in connections depending on $\alpha$) and a correction term. In other words, we have a map $\delta:C^1(\hat{P};\hat{\mathfrak{g}})\rightarrow \operatorname{Conn}(P;\mathfrak{g})$, which is part of the data defining an extension, such that $\hat{S}^{\hat{G}}\circ \jmath= S^G \circ \delta +C$, where $\jmath$ is the inclusion of the correction subspace in the extended domain. Actually, in \cite{category_extension} the authors considered the extensions above not for YM theories, but actually for what they called \textit{Yang-Mills-type} (YMT) theories, which differ from the YM theories in allowing an arbitrary pairing of the curvature in the action funcional. This has the advantage of avoiding the usual requirements of compactness and semi-simplicity on $G$ and the existence of a semi-Riemannian metric on $M$\footnote{Which can suffer of topological obstructions \cite{lorentzian} that become very strong if coupled with other geometric requirements \cite{EHP_yuri_rodney,martins2019geometric}.}. 
For a fixed extension $\hat{G}$ of the gauge group $G$ we considered the space $\operatorname{Ext}(S^G;\hat{G})$ of extensions of a given YMT theory $S^G$ in the above sense. It was shown that for every group $\mathbb{G}$ acting on $\mathbb{R}$ and every commutative unital ring $R$, the space $\operatorname{Ext}(S^G;\hat{G})$ has an induced structure of continuous $R[\mathbb{G}]$-module bundle. Therefore, in view of the classification problem described above, the main strategy is to try a direct topological classification of this bundle. For instance, in \cite{category_extension} the authors conjectured that for every $S^G$ and every $\hat{G}$ the bundle of extensions can be regarded as a subbbundle of a trivial bundle. We notice, however, that another approach is to look at particular subclasses $\mathcal{E}(S^G;\hat{G})\subset \operatorname{Ext}(S^G;\hat{G})$ and try to get some information about it. This is the approach adopted in this paper. More precisely, given a class $\mathcal{E}(S^G;\hat{G})$ of extensions, we propose to consider four questions for it:
\begin{enumerate}
    \item \textit{Existence Problem}: Does it contain nontrivial elements?
    \item \textit{Obstruction Problem}: For an arbitrarily given action functional $\hat{S}$, what are the obstructions to regard it as an element of $\mathcal{E}(S^G;\hat{G})$?
    \item \textit{Maximality Problem}: Is it internally maximal to some other class of extensions?
    \item \textit{Universality Problem}: Is it internally universal to some other class of extensions?
\end{enumerate}

The first two problems are closely related to the classification problem restricted to the class $\mathcal{E}(S^G;\hat{G})$. Indeed, recalling that some trivial extensions exist for any $S^G$ and $\hat{G}$ \cite{category_extension}, the existence problem is about identifying if $\mathcal{E}(S^G;\hat{G})$ contains \textit{at least one} nontrivial extension. Assuming that the existence problem has a positive solution, the obstruction problem is about the study of the space of \textit{all} these nontrivial extensions. 

The other two problems (the maximality problem and the universality problem) are more refined. They are about searching for extensions in the class $\mathcal{E}(S^G;\hat{G})$ which are not only nontrivial, but actually special in some sense. Indeed, we say that $\mathcal{E}(S^G;\hat{G})$ is \textit{internally maximal} relative to other class $\mathcal{E}'(S^G;\hat{G})$ if there is a map $\mu:\mathcal{E}(S^G;\hat{G}) \rightarrow \mathcal{E}'(S^G;\hat{G})$, corresponding to a way to internalize one class into the other, meaning that for each other function $\nu:\mathcal{E}(S^G;\hat{G}) \rightarrow \mathcal{E}'(S^G;\hat{G})$, every extension $\hat{S}^{\hat{G}}\in \mathcal{E}(S^G;\hat{G})$ is such that $\mu(\hat{S}^{\hat{G}})$ can be regarded as an extension of $\nu(\hat{S}^{\hat{G}})$. Thus, this means that $\mu(\hat{S}^{\hat{G}})$ is the maximal way to regard $\hat{S}^{\hat{G}}$ as an element of $\mathcal{E}'(S^G;\hat{G})$. Notice that $\mu(\hat{S}^{\hat{G}})$ is required to be just an extension of $\nu(\hat{S}^{\hat{G}})$, which a priori can be realized in many different ways. If for every $\hat{S}^{\hat{G}}$ there is a unique way to view $\mu(\hat{S}^{\hat{G}})$ as an extension of $\nu(\hat{S}^{\hat{G}})$ we say that $\mathcal{E}(S^G;\hat{G})$ is not only internally maximal, but \textit{internally universal}, relative to $\mathcal{E}'(S^G;\hat{G})$.

Notice that by the definition of maximality and universality, it is suggestive to believe that the fact of a class $\mathcal{E}(S^G;\hat{G})$ being maximal or universal depends crucially on the ambient class $\mathcal{E}'(S^G;\hat{G})$. We will prove, however, that maximality and universality are \textit{dense} properties. If $X$ is a set and $\mathcal{P}(X)$ is its power set, we say that a property $\rho$ is \textit{dense in $X$} if there exists an element $x\in X$ such that $\rho$ holds in every $X_0\cup {x}$ for every $X_0 \in \mathcal{P}(X)$. Equivalently, if there exists $x\in X$ such that if $x\in X_0$, then $\rho$ holds in $X_0 \in \mathcal{P}(X)$\footnote{The relation between this notion of density and the usual notion from topology will be discussed on Section \ref{sec_maximality_universality}.}. Thus, our first main result is the following:$\underset{\;}{\;}$

\noindent \textbf{Theorem A.} \textit{Given a YMT theory $S^G$ and an extension $G\hookrightarrow \hat{G}$ of $G$, there exists an extension $\hat{S}_0^{\hat{G}}$ of $S^G$ such that if it belongs to a class $\mathcal{E}'(S^G;\hat{G})$, then every other class $\mathcal{E}(S^G;\hat{G})$ is internally maximal and universal relative to $\mathcal{E}'(S^G;\hat{G})$.}$\underset{\;}{\;}$

The theorem above depends crucially on the fact that the extensions of a YMT theory $S^G$ constitute a category $ \mathbf{Ext}(S^G;\hat{G})$, as described in \cite{category_extension}. In fact, $\hat{S}_0^{\hat{G}}$ is precisely the terminal object of this category, which is isomorphic to a certain null-type extension and, therefore, to a trivial extension. Although this does not contradict the existence problem, it would be interesting to have maximality and universality criteria independent of trivial extensions. We prove that if we assume the Axiom of Choice in the form of Zorn's Lemma, then such a criterion exists and is obtained by replacing the presence of null-type extensions with certain coherence conditions on the ambient class $\mathcal{E}'(S^G;\hat{G})$, which produces constraints on the size of the classes which are internally maximal/universal to it. More precisely, we prove the following result: $\underset{\;}{\;}$

\noindent \textbf{Theorem B.} \textit{Given a YMT theory $S^G$ and an extension $G\hookrightarrow \hat{G}$ of $G$, if a class of extensions $\mathcal{E}'(S^G;\hat{G})$ is $I$-coherent for maximality (resp. for universality), then every other class $\mathcal{E}(S^G;\hat{G})$ satisfying $\vert \mathcal{E}(S^G;\hat{G})\vert \leq \vert I\vert $ is internally maximal (resp. universal) to $\mathcal{E}'(S^G;\hat{G})$.}$\underset{\;}{\;}$    

Notice that the collection $\operatorname{Ext}(S^G;\hat{G})$ of all extensions of a YMT theory $S^G$ can be partitioned into two disjoint classes: those extensions whose correction term is null, i.e, such that $C\equiv 0$, and those extensions such that $C\neq 0$. Following the terminology of \cite{category_extension}, we say that the first class is of the \textit{complete extensions}, while the second one is of the \textit{incomplete extensions}. In this article we give a special focus on classes of complete extensions, leaving the class of incomplete extensions to a parallel work \cite{gauge_breaking_YMT_extensions}. We considered the class $\operatorname{Coh}(S^G;\hat{G})$ of \textit{almost coherent extensions}, consisting of those extensions of $S^G$ which are complete, whose extended action functional becomes injective on the quotient space of $\widehat{\operatorname{Conn}}(\hat{P};\hat{\mathfrak{g}})$ and which can be built from a categorical pullback. We give a complete description of the existence problem and of the obstruction problem for this class and we prove that if $I$ is small enough, then $\operatorname{Coh}(S^G;\hat{G})$ is $I$-coherent for maximality. In fact, we prove:$\underset{\;}{\;}$

\noindent \textbf{Theorem C.} \textit{Let $S^G$ be a YMT theory and let $\imath: G\hookrightarrow \hat{G}$ be an extension of $G$. Every extension of $S^G$ induces an almost coherent extension of $S^G$ which is nontrivial if the former is. Let $I\subset \operatorname{ED}(S^G;\hat{G})$ be any subset of extended domains which are mutually disjoint. Then $\operatorname{Coh}(S^G;\hat{G})$ is $I$-coherent for maximality, so that every class $\mathcal{E}(S^G;\hat{G})$ with cardinality at least $\vert I\vert$ is internally maximal to $\operatorname{Coh}(S^G;\hat{G})$.$\underset{\;}{\;}$}

The paper is organized as follows. In Section \ref{sec_YMT} we review the basic definitions and the basic facts of \cite{category_extension} concerning YMT theories, the space of extensions and the category of extensions. In Section \ref{sec_extension_problem} the four extension problems described above are presented in more detail and we discuss how to rewrite them in purely categorical language. This is a crucial step in proving Theorems A and B, which is the main objective of Section \ref{sec_maximality_universality}. Finally, in Section \ref{sec_concrete_situation} we introduce the class of almost coherent extensions and we prove Theorem C by means of several steps.

\section{Background}\label{sec_YMT}
\quad\;\, As defined in \cite{category_extension} and following the differential geometry conventions of \cite{kobayashi1996foundations,EHP_yuri_rodney,martins2019geometric}, a Yang-Mills-type theory over a $n$-dimensional smooth manifold $M$ is given by: 
\begin{enumerate}
    \item a real or complex Lie group $G$, regarded as the gauge group of internal symmetries;
    \item a principal $G$-bundle $\pi:P\rightarrow M$, called \emph{instanton sector} of the theory, and whose automorphism group, called the \emph{group of global gauge transformations} is denoted $\operatorname{Gau}_G(P)$;
    \item a $\mathbb{R}$-linear pairing $\langle\cdot,\cdot\rangle: \Omega_{he}^2(P,\mathfrak{g})_\mathbb{R}\otimes_{\mathbb{R}} \Omega_{he}^2(P,\mathfrak{g})_{\mathbb R}\rightarrow \operatorname{Dens}(M)_{\mathbb{R}}$.
\end{enumerate}

Above, $\Omega_{he}^k(P,\mathfrak{g})$ denotes the vector space of horizontal and $G$-equivariant $\mathfrak{g}$-valued $k$-forms in $P$ regarded as a $C^{\infty}(M)$-module via the restriction of scalars by $\pi^*:C^{\infty}(M)\rightarrow C^{\infty}(P)$, $\mathfrak{g}$ is the Lie algebra of $G$ and $\operatorname{Den}(M)$ is the space of 1-densities on $M$, also regarded as a $C^{\infty}(M)$-module. Furthermore, $V_{\mathbb{R}}$ will denote scalar restriction by the inclusion $\mathbb{R} \hookrightarrow C^{\infty}(M)$.  The \emph{action functional} of a YMT over $M$ is the functional 
$$S^G:\operatorname{Conn}(P;\mathfrak{g})\rightarrow \mathbb{R} \quad \text{given by} \quad S^G(D)=\int_M\langle F_D,F_D\rangle,$$ where $\operatorname{Conn}(P;\mathfrak{g})\subset \Omega ^1(P;\mathfrak{g})$ is the set of connection 1-forms on $P$ and $F_D=d_DD$ is the field-strength of $D$, where $d_D:\Omega^k(P;\mathfrak{g})\rightarrow \Omega_{he}^{k+1}(P;\mathfrak{g})$ is the \emph{exterior covariant derivative} of $D$, such that $d_D\alpha = d\alpha + D[\wedge]_{\mathfrak{g}}\alpha$. Following \cite{EHP_yuri_rodney,martins2019geometric}, $[\wedge]_\mathfrak{g}$ denotes the wedge product induced by the Lie bracket of $\mathfrak{g}$. Without risk of confusion, it will be denoted simply by $[\wedge]$.

\begin{example}
\emph{Every Yang-Mills theory over a semi-Riemannian manifold $(M,g)$ is a YMT theory with the pairing $\langle \alpha,\beta \rangle = \operatorname{tr}( \alpha [\wedge] \star_g \beta)$. As shown in \cite{category_extension}, for every triple $(M,G,P)$ the corresponding collection $\operatorname{YM}_G(P)$ of YMT theories over $M$ which has $G$ as the gauge group and $P$ as the instanton sector is a real vector space. If $G$ is discrete or $\dim M<2$, then this space is zero-dimensional, implying that under such conditions this YMT theory is a YM theory. Otherwise it is infinite-dimensional, meaning that when working with YMT theories instead of only with YM theories we will have an infinity amount of new degrees of freedom to work with.}
\end{example}
\begin{remark}
In some cases, such as for those discussed in \cite{category_extension}, it is better to work with a more restricted class of YMT theories. Indeed, we say that $S^G$ is \textit{tensorial} (resp. \textit{gauge invariant}) if its pairing $\langle \cdot, \cdot \rangle$ is $C^{\infty}(M)$-linear (resp. invariant by the action of $\operatorname{Gau}_G(P)$, i.e,  $\langle f^{*}\left(F_D\right),f^{*}\left(F_D\right)\rangle = \langle F_D,F_D\rangle$ for every $f\in \operatorname{Gau}_G(P)$). In the following, however, we will work in full generality. Notice that by avoiding the gauge invariance hypothesis on $\langle \cdot, \cdot \rangle$ we are allowing that $S^G$ be gauge breaking in the classical sense. Such hypotheses of linearity and gauge invariance are important for doing local computations, as in \cite{gauge_breaking_YMT_extensions}.
\end{remark}

\subsection{Extensions}\label{sec_extension}

\quad\;\, Let $G$ be a Lie group. A \emph{basic extension} of $G$ is another Lie group $\hat{G}$ and an injective Lie group homomorphism $i:G\rightarrow \hat{G}$.
\begin{example}

\emph{Every split Lie group extension induces a basic extension. Actually, a Lie group extension is a basic extension iff it is split \cite{category_extension}. On the other hand, there are many basic extensions which are not Lie group extensions.}
\end{example}
Let $P\rightarrow M$ be a principal $G$-bundle. Each basic extension $\imath : G\hookrightarrow\hat{G}$  induces a $\hat{G}$-bundle $\hat{P}$, which becomes equipped with a monomorphism $\imath: P\hookrightarrow\hat{P}$ and a group homomorphism $\xi:\operatorname{Gau}_{G}(P)\rightarrow \operatorname{Gau}_{\hat{G}}(\hat{P})$, as described in \cite{advanced_field_theory,xi_homotopy,category_extension}. Let $S^G$ be a YMT theory with gauge group $G$ and instanton sector $P$. As defined in \cite{category_extension}, an \emph{extension} of $S^G$ relative to a basic extension $\imath: G \hookrightarrow \hat{G}$ of $G$ is given by:
\begin{enumerate}
    \item a $\operatorname{Gau}_{\hat{G}}(\hat{P})$-invariant subspace of equivariant 1-forms in $\hat{P}$, called \textit{extended domain}, such that: \begin{equation}{\label{domain_extension}
    \operatorname{Conn}(\hat{P};\hat{\mathfrak{g}}) \subseteq \widehat{\operatorname{Conn}}(\hat{P};\hat{\mathfrak{g}})\subseteq \Omega^{1}_{eq}(\hat{P};\hat{\mathfrak{g}}), }
    \end{equation}
    \item a $\operatorname{Gau}_{\hat{G}}(\hat{P})$-invariant functional $\hat{S}^{\hat{G}}:\widehat{\operatorname{Conn}}(\hat{P};\hat{\mathfrak{g}})\rightarrow \mathbb{R}$, called the \emph{extended action functional};
    \item a subspace $0\subseteq C^1(\hat{P};\hat{\mathfrak{g}})\subseteq \widehat{\operatorname{Conn}}(\hat{P};\hat{\mathfrak{g}})$, called the  \emph{correction subspace};
    \item a \emph{correction functional} $C:C^1(\hat{P};\hat{\mathfrak{g}})\rightarrow \mathbb{R}$ and map $\delta: C^1(\hat{P};\hat{\mathfrak{g}})\rightarrow \operatorname{Conn}(P;\mathfrak{g})$ such that
    \begin{equation}\label{decomposition_extension}
      \hat{S}^{\hat{G}}{\big|}_{C^1(\hat{P};\hat{\mathfrak{g}})} = S^{G}\circ \delta + C.
    \end{equation}
\end{enumerate}
\begin{remark}
As discussed in the Introduction and in Section 3 of \cite{category_extension}, this definition was proposed as an unification for three kinds of  ``extensions of YM theories'' arising in the literature: extensions by deformations, extensions by adding correction terms and extensions by augmenting the gauge group.
\end{remark}
\begin{remark}
Let $\operatorname{Ext}(S^G;\mathfrak{g})$ be the collection of all extensions of a given YMT theory $S^G$ relative to a fixed basic extension $\imath:G\hookrightarrow \hat{G}$ and let  $\operatorname{ED}(S^G;\hat{G})$ be the space of all extended domains, i.e, of all $\operatorname{Gau}_{\hat{G}}(\hat{P})$-invariant sets of 1-forms in $\hat{P}$ satisfying (\ref{domain_extension}). In \cite{category_extension} it was proven that the obvious projection $\pi: \operatorname{Ext}(S^G;\mathfrak{g})\rightarrow \operatorname{ED}(S^G;\hat{G})$ is a bundle of commutative monoids. Furthermore, each additive action of a group $\mathbb{G}$ on $\mathbb{R}$ lifts the above structure of a bundle of commutative monoids to a structure of $R[\mathbb{G}]$-module bundle, where $R[\mathbb{G}]$ is the group $R$-algebra and $R$ is any commutative unital ring.
\end{remark}

Fixing extensions $\hat{S}_i^{\hat{G}}\in \operatorname{Ext}(S^G;\hat{G})$, with $i=1,2$, a \textit{morphism} between them is a pair of functions $(f,g)$ such that $f:\widehat{\operatorname{Conn}}_1(\hat{P};\hat{\mathfrak{g}})\rightarrow \widehat{\operatorname{Conn}}_2(\hat{P};\hat{\mathfrak{g}})$ and $g:C^1_1(\hat{P};\hat{\mathfrak{g}})\rightarrow C^1_2(\hat{P};\hat{\mathfrak{g}})$, and such that the underlying structures are preserved, i.e, such that $f$ is $\operatorname{Gau}_{\hat{G}}(\hat{P})$-equivariant and the diagram below commutes. This defines a category $\mathbf{Ext}(S^G;\hat{G})$ which can be realized as a conservative subcategory of a product of slice categories \cite{category_extension}.
\begin{equation}{\label{diagram_equivalence_extensions}
\xymatrix{   && \ar@/_/[lld]_-{\delta _1} \ar[dd]^g C^1_1(\hat{P};\hat{\mathfrak{g}}) \ar[ld]_{C_1} \ar@{^(->}[rr] && \widehat{\operatorname{Conn}}_1(\hat{P};\hat{\mathfrak{g}}) \ar[ld]_-{\hat{S}^{\hat{G}}_1} \ar[dd]^f \\
\operatorname{Conn}(P;\mathfrak{g}) & \mathbb{R} & & \mathbb{R} \\
 && \ar@/^/[llu]^-{\delta _2} \ar[lu]^{C_2} C^1_2(\hat{P};\hat{\mathfrak{g}}) \ar@{^(->}[rr] &&   \widehat{\operatorname{Conn}}_2(\hat{P};\hat{\mathfrak{g}}) \ar[lu]^-{\hat{S}^{\hat{G}}_2}}}
\end{equation}

\begin{remark}\label{remark_null_type}
For every YMT $S^G$ and every basic extension $\imath :G\hookrightarrow \hat{G}$ the category $\mathbf{Ext}(S^G;\hat{G})$ has a terminal object given by the null extension such that $\widehat{\operatorname{Conn}}(\hat{P};\hat{\mathfrak{g}})=\Omega ^1_{eq}(\hat{P};\hat{\mathfrak{g}})$, $C^1(\hat{P};\hat{\mathfrak{g}})=0$ and such that the maps $\hat{S}^{\hat{G}}$, $C$ and $\delta$ are all null. See Example 2 and Example 15 of \cite{category_extension}.
\end{remark}
\begin{remark}\label{remark_trivial_extensions}
The null extension is an example of a \textit{trivial extension} which exists for any YMT theory $S^G$. In addition to the null extensions (where the maps involved are null) the trivial extensions also include the identity extensions and the constant extensions, whose maps are all identities or constants, respectively. See Section 3.1 of \cite{category_extension}. Extensions which are not of null-type, constant-type or identity-type are called \textit{nontrivial}. In Section 3.4 of \cite{category_extension} it was shown that emergence phenomena (in the sense of \cite{emergence_yuri}) play an important role in building nontrivial extensions.
\end{remark}

\section{The Extension Problems}\label{sec_extension_problem}

\quad\;\, Let $S^G$ be a YMT and let $\imath:G\hookrightarrow \hat{G}$ be a basic extension of $G$. A \textit{class} of extensions of $S^G$, relative  to $\imath:G\hookrightarrow \hat{G}$, is just a subset $\mathcal{E}(S^G;\hat{G})$ of $\operatorname{Ext}(S^G;\hat{G})$. 
\begin{remark}
Throughout the paper we will always consider classes of extensions of the same YMT $S^G$ and relative to the same basic extension $\hat{G}$ of $G$. Therefore, in order to simplify the notation, when $S^G$ and $\hat{G}$ are implicitly understood we will write $\mathcal{E}$ instead of $\mathcal{E}(S^G;\hat{G})$ to denote a class of extensions. 
\end{remark}

The basic problems concerning classes of extensions of a given YMT theory $S^G$ relative to a basic extension $\imath:G\hookrightarrow$ are the following. 

\begin{existence}\label{question_existence}
Let $\mathcal{E}$ be a  class of extensions. Under which conditions it has nontrivial elements in the sense of Remark \ref{remark_trivial_extensions}?
\end{existence}

\begin{obstruction}\label{question_obstruction}
Let $\hat{S}$ be an action funcional. What are the obstructions to realize it as an element of $\mathcal{E}$?
\end{obstruction}

Besides these basic questions, there are others which are more refined. Let $\operatorname{Clss}(S^G;\hat{G})$ be the collection of all classes of extensions of $S^G$ relative to $\imath:G\hookrightarrow \hat{G}$, i.e, the power set of $\operatorname{Ext}(S^G;\hat{G})$, and notice that we have a relation of \textit{maximality} on it, denoted by ``$\hookrightarrow_{\operatorname{max}}$''. Indeed, we say that a class $\mathcal{E}_0$ is \textit{internally maximal} to other class $\mathcal{E}_1$, and we write $\mathcal{E}_0 \hookrightarrow_{\operatorname{max}} \mathcal{E}_1$, if there exists a function $\mu:\mathcal{E}_0\rightarrow \mathcal{E}_1$ satisfying the following condition: for every other function $\nu:\mathcal{E}_0\rightarrow \mathcal{E}_1$ and every $\hat{S}^{\hat{G}}\in \mathcal{E}_0$, we can regard $\mu(\hat{S}^{\hat{G}})$ as an extension of $\nu(\hat{S}^{\hat{G}})$.  We can also consider another relation, called \textit{universality relation}, denoted by $\hookrightarrow_{\operatorname{unv}}$, in which  $\mathcal{E}_0 \hookrightarrow_{\operatorname{unv}} \mathcal{E}_1$ iff $\mathcal{E}_0$ is internally maximal to $\mathcal{E}_1$ in an universal way, meaning that  $\mu(\hat{S}^{\hat{G}})$ can be regarded as an extension of $\nu(\hat{S}^{\hat{G}})$ in a \textit{unique} way. 

\begin{maximality}
Given two classes of extensions of $S^G$, under which conditions one of them is internally maximal to the other? In other words, what are the classes of $\hookrightarrow_{\operatorname{max}}$?
\end{maximality}
\begin{universality}
Given two classes of extensions of $S^G$, under which conditions one of them is internally universal to the other? In other words, what are the classes of  $\hookrightarrow_{\operatorname{unv}}$?
\end{universality}

Concretely, when the maximality problem has a solution, this means that every extension of $S^G$ belonging to $\mathcal{E}_0$ can be regarded as an element of $\mathcal{E}_1$ in such a way that it itself admits extensions. Not only this, it admits a special extension which is maximal. The physical interpretation of the universality problem is analogous, but now the special extension can be regarded as universal among all other extensions.

Notice, however, that the relations $\hookrightarrow_{\operatorname{max}}$ and $\hookrightarrow_{\operatorname{unv}}$ are transitive, but typically not reflexive, symmetric or antisymmetric. Therefore, the decomposition of $\operatorname{Ext}(S^G;\hat{G})$ into their classes is very complex, which makes the direct analysis of the maximality and universality problems very difficult. Recalling that $\operatorname{Ext}(S^G;\hat{G})$ is the collection of objects of the category $\mathbf{Ext}(S^G;\hat{G})$ and that the notion of ``extension $\hat{S}_2^{\hat{G}}$ of an extension $\hat{S}_1^{\hat{G}}$'' of $S^G$, used to describe the maximality and universality conditions above, can be interpreted by the existence of a morphism $\hat{S}_1^{\hat{G}}\rightarrow \hat{S}_2^{\hat{G}}$ in $\mathbf{Ext}(S^G;\hat{G})$, we conclude that the maximality and universality problems is the study of decompositions of the the collection of objects of a category which are induced by relations determined by morphisms between the objects. 

Since for every category $\mathbf{C}$ there are natural and more well-behaved relations in $\operatorname{Ob}(\mathbf{C})$ (e.g, which are partial order relations and equivalence relations), it is natural to ask if the problems described above cannot be rewritten in purely categorical terms, such that the bad-behaved relations $\hookrightarrow_{\operatorname{max}}$ and  $\hookrightarrow_{\operatorname{unv}}$ become determined by these natural and well-behaved relations. Doing this is the main objective of the remaining of this section.

\begin{remark}[The empty case]\label{remark_empty}
Let $\mathcal{E}=\varnothing $ be the empty class. Then the existence and obstruction problems have obvious negative solutions. Furthermore, by vacuous truth this class is internally universal (and in particular maximal) to any other. Thus, from now on we will consider only nonempty classes of extensions.
\end{remark}

\subsection{Rewriting the Extension Problems}

\quad\;\,Recall that the full subcategories of a given category $ \mathbf{C}$ are in bijection with the subclasses of the class of objects $\operatorname{Ob}(\mathbf{C})$ of $\mathbf{C}$. Thus, to give a class $\mathcal{E}$ of extensions whose underlying basic extension is $\imath:G\hookrightarrow \hat{G}$ is equivalent to giving a full subcategory of $\mathbf{Ext}(S^G;\hat{G})$, which in the following will be denoted by the same notation. We will say that a full subcategory of $\mathbf{Ext}(S^G;\hat{G})$ is \textit{nontrivial} if its class of objects is nontrivial in the sense of Remark \ref{remark_trivial_extensions}, i.e, if it does not contain the null-type, identity-type and constant-type extensions. With this, the existence problem and the obstruction problem can be rewritten in the following equivalent way: 

\begin{existence}[categorical form]
Is a given full subcategory $\mathcal{E}$ of $\mathbf{Ext}(S^G;\hat{G})$ nontrivial? 
\end{existence}
\begin{obstruction}[categorical form]
What are the subcategories of $\mathbf{Ext}(S^G;\hat{G})$ which are adjoint to a given full subcategory $\mathcal{E}$?
\end{obstruction}

The direct categorical analogues of the maximality and universality conditions described in the last section are the following.  We say that a full subcategory $\mathbf{C}_0$ of $\mathbf{C}$ is \textit{weakly internally maximal} to other full subcategory $\mathbf{C}_1$ if there exists a functor $\mu:\mathbf{C}_0\rightarrow\mathbf{C}_1$ such that for every other functor $\nu:\mathbf{C}_0\rightarrow\mathbf{C}_1$ there exists a natural transformation $u:\nu \Rightarrow \mu$. If this natural transformation is unique, we say that $\mathbf{C}_0$ is \textit{weakly universal}. Thus:

\begin{maximality}[weak categorical form] Given two full subcategories of $\mathbf{Ext}(S^G;\hat{G})$, are they weakly internally maximal one to the other? 
\end{maximality}

\begin{universality}[weak categorical form] Are two given full subcategories of $\mathbf{Ext}(S^G;\hat{G})$ weakly internally universal one to the other? 
\end{universality}

Although the direct categorical analogues of the existence and universality problems are \textit{equivalent} to the original problems, we note that the above direct categorical analogues of the maximality and universality problems are a bit \textit{different} of those in Subsection \ref{sec_extension_problem}, which is the reason why we used the adjective ``weakly''. Indeed, notice that every functor $F:\mathbf{C}_0\rightarrow \mathbf{C}_1$ induces a map between the underlying sets of objects and other map between the sets of morphisms. But in the original universality problem we only require information between objects. Therefore, if the map on objects lifts uniquely to a functor, then the weak maximality and weak universality imply maximality and universality in the original sense. 

The problem is that a function $f:\operatorname{Ob}(\mathbf{C}_0)$ does not need extending to a functor $F:\mathbf{C}_0\rightarrow \mathbf{C}_1$. Furthermore, even if such lifts exist, they are typically highly non-unique. Indeed, let $\operatorname{Aut}(\operatorname{img}(f))$ be the disjoint union of the automorphism groups  $\operatorname{Aut}(f(X))$ of $f(X)\in \mathbf{D}$ for every $X\in\mathbf{C}$ and consider the obvious projection $\pi:\operatorname{Aut}(\operatorname{img}(f)) \rightarrow \operatorname{Ob}(\mathbf{C})$. Each section for this projection can be used to modify the functor $F$ lifting $f$ only on morphisms. Furthermore, the modified functor is not equal to $F$ iff the section is nontrivial, which means that the automorphism groups $\operatorname{Aut}(f(X))$ are nontrivial. Only in very rare cases can we ensure that any two functors lifting $f$ are naturally isomorphic \cite{functors_objects}. 

Thus, if the basic obstruction in identifying maps on objects with functors (and therefore the weakly maximality/universality problem with the original maximality/universality problem) are the automorphisms, the idea is to work up to automorphisms.

\subsection{Maximality and Universality Up to Isomorphisms}\label{maximality_up_to_isomorphisms}

 \quad\;\, Let $\mathbf{C}$ be a category and let $\operatorname{Iso}(\mathbf{C})$ denote the quotient space $\operatorname{Ob}(\mathbf{C})/\simeq$ of objects by the isomorphism relation $\simeq$. Notice that the relation ``\emph{there exists a morphism between two objects $X$ and $Y$}'' is a partial order $\leq$ on $\operatorname{Ob}(\mathbf{C})$. It does not need to have a maximal element. Let $\mathbf{C}_0\subset \mathbf{C}$ be a full subcategory, so that we have an induced partial order on $\operatorname{Ob}(\mathbf{C}_0)\subset \operatorname{Ob}(\mathbf{C})$. It is compatible with the isomorphism relation, meaning that if $X\simeq X'$ and $Y\simeq Y'$, then $X\leq Y$ iff  $X'\leq Y'$. Thus, we have an induced partial order $\lesssim$ on $\operatorname{Iso}(\mathbf{C})$. We say that $\mathbf{C}_0$ is a \emph{greatest subcategory} if $(\operatorname{Iso}(\mathbf{C}_0),\lesssim)$ has a greatest element, i.e, if there exists an object $X_0\in \mathbf{C}_0$ and a morphism $Y\rightarrow X_0$ for each other $Y$. 
 
 Now, let $\mathbf{C}_0$ and $\mathbf{C}_1$ be full subcategories of $\mathbf{C}$. We say that $\mathbf{C}_0$ is \textit{internally maximal up to isomorphism} to $\mathbf{C}_1$ if there exists a map on objects $f:\operatorname{Ob}(\mathbf{C}_0) \rightarrow \operatorname{Ob}(\mathbf{C}_1)$ such that the full subcategory defined by its image is a greatest subcategory of $\mathbf{C}_1$.
 
 \begin{maximality}[categorical form] Given two full subcategories of $\mathbf{Ext}(S^G;\hat{G})$, is one internally maximal up to isomorphisms to the other?
 \end{maximality}
 
 Given full categories $\mathbf{C}_0,\mathbf{C}_1$ of $\mathbf{C}$,  notice that $\mathbf{C}_0$ is internally maximal up to isomorphisms to $\mathbf{C}_1$  iff there exists $Y_0\in \mathbf{D}$ such that $Y_0\simeq f(X_0)$ for some $X_0\in \mathbf{C}$ and such that there exists a morphism $Y\rightarrow Y_0$ for each other $Y\in \mathbf{D}$ such that $Y\simeq f(X)$ for some $X\in \mathbf{C}$. Due to the compatibility between $\simeq$ and $\leq$, the previous condition is also equivalent to the folllowing: there exists $X_0 \in \mathbf{C}_0$ such that for each other $X\in \mathbf{C}_0$ there exists a morphism $f(X)\rightarrow f(X_0)$. But this is exactly the maximality condition of Subsection  \ref{sec_extension_problem}. Thus, the classical and the categorical forms of the maximality problem are equivalent. Explicitly, we have proved the following:
 \begin{proposition}\label{prop_maximality}
     Let $S^G$ be a YMT with gauge group $G$ and let $\imath:G\hookrightarrow \hat{G}$ be a basic extension of $G$. A class $\mathcal{E}_0$ of extensions of $S^G$ is internally maximal to other class $\mathcal{E}_1$ iff the full subcategory defined by $\mathcal{E}_0$ is internally maximal up to isomorphisms to the subcategory defined by $\mathcal{E}_1$.
 \end{proposition}

If we try to do a similar analysis with the universality condition we find a problem. Indeed, we should replace the relation \emph{``there exists a monomorphism between $X$ and $Y$''} with the relation \emph{``there exists a \textbf{unique} morphism between $X$ and $Y$''}. As we saw, the first one is a partial relation which behaves very well relative to isomorphisms. But the second one is typically not a partial relation precisely because isomorphisms are involved. Actually, it defines a partial relation in $\mathbf{C}$ iff $\mathbf{C}$ is a category without isomorphisms\footnote{A category without isomorphisms is also called a \textit{gaunt category}. They appear in abstract homotopy theory and higher category theory as the class of categories whose simplicial nerve defines a complete Segal space \cite{gaunt_category}.}, i.e, if a morphism $f:X\rightarrow Y$ is not an isomorphism unless $Y=X$ and $f=id_X$. But this is a very strong assumption. For instance, it implies that the automorphism group $\operatorname{Aut}(X)$ is trivial for every $X\in\mathbf{C}$. Furthermore,  $\mathbf{C}$ cannot be the codomain of a functor $F:\mathbf{D}\rightarrow \mathbf{C}$ where  $\mathbf{D}$ has at least one isomorphism $f:X\rightarrow Y$ with $F(X)\neq F(Y)$. Even so, if we assume this hypothesis we can apply the discussion to the universality condition.

Thus, let $\mathbf{C}_0$ and $\mathbf{C}_1$ be full subcategories $\mathbf{C}$ and assume that $\mathbf{C}_1$ is without isomorphisms. Let $\leqq$ be the partial relation in $\operatorname{Ob}(\mathbf{C}_1)$ such that $X\leqq Y$ iff there exists a unique morphism $X\rightarrow Y$. Since the full subcategories are without isomorphisms, we have $\operatorname{Ob}(\mathbf{C}_1)\simeq  \operatorname{Iso}(\mathbf{C}_1)$, with $i=0,1$ so that $\leqq$ trivially pass to quotient and define another partial relation $\lessapprox$. We say that $\mathbf{C}_0$ is \textit{internally universal up to isomorphisms} to $\mathbf{C}_1$ if there is a map on objects $f:\operatorname{Ob}(\mathbf{C}_0)\rightarrow \operatorname{Ob}(\mathbf{C}_1)$ such that $\operatorname{img}(f)\subset \operatorname{Ob}(\mathbf{C}_1)$ has a greatest element when considered as a subset of $(\operatorname{Iso}(\mathbf{C}_1),\lessapprox)$. Since $\mathbf{C}_1$ is without isomorphisms, it happens that $\operatorname{img}(f)$ has a greatest element in $(\operatorname{Ob}(\mathbf{C}_1),\leqq)$.

\begin{universality}[categorical form]
Given two full subcategories $\mathcal{E}_0$ and $\mathcal{E}_1$ of $\mathbf{Ext}(S^G;\hat{G})$, where the second one is without isomorphisms, under which conditions $\mathcal{E}_0$ is internally universal up to isomorphisms to $\mathcal{E}_1$?
\end{universality}
 
Under the hypothesis that $\mathcal{E}_1$ is without isomorphisms, one can check that Proposition \ref{prop_maximality} remains true if we replace ``maximality'' with ``universality''. More precisely, one can check that the solutions of the above problem correspond precisely to extensions $\hat{S}^{\hat{G}}_0 \in \mathcal{E}_0$ such that there exists a map $\mu:\mathcal{E}_0\rightarrow \mathcal{E}_1$ with the property that for each other extension $\hat{S}^{\hat{G}}$ there exists a unique morphism $\mu(\hat{S}^{\hat{G}})\rightarrow \mu(\hat{S}^{\hat{G}}_0)$.

\begin{remark}
Since by Proposition \ref{prop_maximality}  maximality and maximality up to isomorphisms are equivalent conditions, in the following we will refer to both as ``\textit{maximality}''. Since universality and universality up to isomorphisms are conditions which are not comparable when $\mathcal{E}_1$ is not without isomorphisms, we will mantain the distinction between them.
\end{remark}

\section{Conditions for  Maximality and Universality}\label{sec_maximality_universality}

\quad\;\, After the prevous categorical characterizations of the fundamental problems involving extensions of YMT theories, we are ready to try to solve them. Here we will present some sufficient conditions to ensure maximality or universality between classes of extensions of YMT theories.  We will first show that if we are analyzing universality internal to some class of extensions which contains the null-type extensions of Remark \ref{remark_null_type}, then universality is immediate in both categorical senses, i.e, weak universality and universality up to isomorphisms.

\begin{theorem}[Theorem A]\label{thm_universality_dense}
Let $\mathcal{E}_1$ be a class of extensions of $S^G$ relative to $\imath:G\hookrightarrow \hat{G}$.
\begin{enumerate}
    \item  If it contains the null-type extension of Remark \ref{remark_null_type} or any extension isomorphic to it, then each other class $\mathcal{E}_0$ is internally universal and weakly universal (in particular maximal and weakly maximal) to $\mathcal{E}_1$;
    \item If in addition $\mathcal{E}_1$ is a category without isomorphisms, then every $\mathcal{E}_0$ is also universal up to isomorphisms.
\end{enumerate}
\end{theorem}
\begin{proof}
First of all, notice that the second assertion follows directly from the first one by recalling that, by the discussion of Section \ref{maximality_up_to_isomorphisms}, under those hypotheses universality holds iff universality up to isomorphisms holds. Thus, it is enough to work on the first assertion. We affirm that the first hypothesis implies that the functor category $\operatorname{Func}(\mathcal{E}_0;\mathcal{E}_1)$ has a terminal object. Indeed, recall that the null extension of Remark \ref{remark_trivial_extensions} is a terminal object in the category $\mathbf{Ext}(S^G;\hat{G})$ \cite{category_extension}. Since by hypothesis it belongs to $\mathcal{E}_1$, the fact that limits are reflected by fully faithful functors implies that it is actually a terminal object in $\mathcal{E}_1$. Once limits of a functor category are determined by limits of the codomain category, it follows that the null extension defines a terminal object in $\operatorname{Func}(\mathcal{E}_0;\mathcal{E}_1)$. Since objects isomorphic to terminal objects are also terminal, it follows that the conclusion above remains valid if we replace the null extension by any other isomorphic to it. Now, notice that each terminal object in that functor category is a realization of $\mathcal{E}_0$ as weakly internally universal to $\mathcal{E}_1$, proving the first assertion in the weakly universal case. For the universal case, recall that from Remark \ref{remark_empty} if $\mathcal{E}_0$ is empty, then there is nothing to do. Thus, assume $\mathcal{E}_0\neq \varnothing$. Let $\hat{S}^{\hat{G}}_0$ be an extension isomorphic to the null extension (which by hypothesis belongs to $\mathcal{E}_1$) and let $\mu:\mathcal{E}_0 \rightarrow\mathcal{E}_1$ be any map whose image contains $\hat{S}^{\hat{G}}_0$. Let $ \hat{S}^{\hat{G}}_{0,*}\in \mu^{-1}(\hat{S}^{\hat{G}}_0)$ and notice that since $\hat{S}^{\hat{G}}_0$ is a terminal object, for every other $\hat{S}^{\hat{G}}\in \mathcal{E}_0$ there exists a unique morphism $\mu(\hat{S}^{\hat{G}})\rightarrow \mu(\hat{S}^{\hat{G}}_{0,*})$, which proves universality. 
\end{proof}

Let $X$ be a set, $\mathcal{P}(X)$ its power set and $\mathcal{P}^2(X)$ be the power set of its power set. Consider a property $\rho$ in $\mathcal{P}(X)$. We say that $\rho$ is dense at $\mathcal{P}_0(X) \in \mathcal{P}^2(X)$ if there exists an element $x_0\in X$ such that for every $Y\subset X$ in $\mathcal{P}_0(X)$ the property $\rho$ holds in $Y\cup{x_0}$. If $\mathcal{P}_0(X)=\mathcal{P}(X)$ we say simply that $\rho$ is \textit{dense}.  The terminology is due to the following fact: let $\mathcal{P}_\rho(X)$ be the subset of those subsets of $X$ in which property $\rho$ is verified. For each $x_0\in X$ consider the map $\cup _{x_0}: \mathcal{P}(X) \rightarrow \mathcal{P}(X)$ such that $\cup _{x_0}Y=Y\cup {x_0}$ obtained by taking the union with ${x_0}$. Take a topology in $\mathcal{P}_\rho(X)$ for which the map $\cup _{x_0}$ is continuous. If $\mathcal{P}_0(X)$ is dense\footnote{This hypothesis can be avoided if we consider a topology in $\mathcal{P}(X)$ in which evaluation of $\cup$ at points is always dense.}, then its image by $\cup _{x_0}$ is also dense. But $\rho$ being dense at $ \mathcal{P}_0(X)$ means that this image is contained in $\mathcal{P}_{\rho}(X)$, so that $\mathcal{P}_{\rho}(X)$ is dense too.

Given a relation $\mathcal{R}$ in $\mathcal{P}(X)$, we can consider the induced property $\rho_{\mathcal{C}}$  ``\textit{being $\mathcal{R}$-related with a fixed $X_0\subset X$}''. We say that $\mathcal{R}$ is \textit{dense at $\mathcal{P}_0(X)$} (resp. \textit{dense}) if $\rho_{\mathcal{R}}$ is.

\begin{corollary}
For every YMT $S^G$ and every basic extension $\imath:G\hookrightarrow \hat{G}$, both the universality relation, the maximality relation, the weak universality relation and the weak maximality relations are all dense. The universality up to isomorphisms relation is dense in the classes of extensions whose induce full subcategories are without isomorphisms. 
\end{corollary}
\begin{proof}
Straightforward from the above discussion and from Theorem \ref{thm_universality_dense}.
\end{proof}

\subsection{Zorn's Lemma Approach}\label{sec_zorn}

\quad\;\,Let us now give a condition for maximality and universality even in the case when a class of extensions $\mathcal{E}$ does not contains an object isomorphic to the null extension. Recalling that maximality and universality up to isomorphism conditions are about existence of greatest elements, the main idea is to use Zorn's lemma. This is precisely what we will do.

Given a set of indexes $I$ and a class of extensions $\mathcal{E}$, let $\mathcal{E}(I)$ denote the union of all subclasses $X\subset \mathcal{E}$ such that $\vert X \vert \leq \vert I\vert $. Then this is another subclass of extensions of $\mathcal{E}$. We say that $\mathcal{E}$ is $I$-\textit{coherent for maximality} if it has $J$-coproducts for every $J\subset I$ and if the partial order $\lesssim$ of Subsection \ref{maximality_up_to_isomorphisms} is actually a total order when restricted to $\mathcal{E}(I)$. We say that it is \textit{coherent for maximality} if it is $I$-coherent for $I$ containing the class of objects of $\mathcal{E}$. Similarly, we say that it is $I$-\textit{coherent for universality up to isomorphisms} if it has $J$-coproducts, $\mathcal{E}$ is without isomorphisms  and the partial order $\lessapprox$ is a total order.

\begin{example}\emph{
If $\mathcal{E}$ is a coherent category in the sense of \cite{elephant}, then it is clearly coherent for maximality. The product and the slice of coherent categories are coherent categories. The category of sets is coherent. By Theorem 5.1 of \cite{category_extension}, every $\mathcal{E}$ is a conservative subcategory of a product of slices of $\mathbf{Set}$, so that it is a conservative subcategory of a coherent category. Thus, under additional requirements it can itself be made coherent \cite{elephant}.}
\end{example}

\begin{remark}
Notice that if we allow $J=\varnothing$ in the definition of $I$-coherence above, then this implies that the category have initial objects, which is not covered by Theorem \ref{thm_universality_dense}, since there we assumed existence of terminal objects.
\end{remark}

\begin{theorem}[Theorem B]\label{thm_zorn}
Let $\mathcal{E}_1$ be a class of extensions as above. If it is $I$-coherent for maximality (resp. $I$-coherent for universality up to isomorphisms), then every $\mathcal{E}_0$ such that $\vert \mathcal{E}_0\vert \geq \vert \mathcal{E}_1(I) \vert$ is internally maximal (resp. universal up to isomorphisms) to $\mathcal{E}_1(I)$.    
\end{theorem}
\begin{proof}
Consider first the case $\mathcal{E}_0=\operatorname{Ob}(\mathbf{Ext}(S^G;\hat{G}))$.  By the closeness under $J$-coproducts, with $J\subset I$, it follows that each subset $X$ of $\mathcal{E}_1$ such that $\vert X \vert \leq \vert I\vert$ is bounded from above. Indeed, let $b(X)=\coprod_{x\in X}x$, so that for every $x\in X$ we have the canonical morphism $x\rightarrow b(X)$ and therefore $x\lesssim b(X)$. Thus, each subset of  $\mathcal{E}(I)$ is bounded from above. The total order hypothesis on $\lesssim$ implies that each subset is a chain. Therefore, by Zorn's lemma we see that $\mathcal{E}_1(I)$ has a maximal element which is actually a greatest element $\hat{S}^{\hat{G}}_0$. By the Axiom of Choice there exists a function $r:\operatorname{Ob}(\mathbf{Ext}(S^G;\hat{G}))\rightarrow \mathcal{E}_1(I)$ which is a retraction for the reciprocal inclusion $\mathcal{E}_1 \hookrightarrow \operatorname{Ob}(\mathbf{Ext}(S^G;\hat{G}))$ and therefore a surjective map. Since $r$ is surjective, $r^{-1}(\hat{S}^{\hat{G}}_0)$ is nonempty and each element on it satisfies the desired maximality or universality up to isomorphism condition. Finally, note that the same argument holds if $\mathcal{E}_0$ is any class in which $\mathcal{E}_1(I)$ can be embedded, i.e, such that $\vert \mathcal{E}_0 \vert \leq \vert \mathcal{E}_1(I) \vert$.
\end{proof}
\begin{remark}
We can recover a version of Theorem \ref{thm_universality_dense} from the last theorem. Indeed, since $I$-coproducts with $I=\varnothing$ is the same as an initial object and since the null-extension is actually a null object in the category of extensions, it follows that $\mathcal{E}(\varnothing)\simeq *$ is the class containing the null extension. Suppose $\mathcal{E}_0$ nonempty, so that $\vert \mathcal{E}_0\vert \geq 1 = \vert \mathcal{E}(\varnothing)\vert$. Thus, by the last theorem, $\mathcal{E}_0$ is internally maximal to $\mathcal{E}_1(\varnothing)$. But note that if a class $\mathcal{E}_0$ is internally maximal to $\mathcal{E}_1$, then it is also internally maximal to any class containing $\mathcal{E}_1$. Therefore, we conclude that every $\mathcal{E}_0$ is internally maximal to any class containing the null extension, which is essentially Theorem \ref{thm_universality_dense}.
\end{remark}

\section{A Concrete Situation}\label{sec_concrete_situation}

\quad\;\, In this section we will analyze the four problems of Section \ref{sec_extension_problem} for a concrete class of extensions. More precisely, we prove the following, which  basically constitutes Theorem C:
\begin{enumerate}
    \item for every group $G$, every YMT theory $S^G$ and every basic extension $\imath :G\hookrightarrow \hat{G}$, the class $\operatorname{Inj}(S^G;\hat{G})$ of \emph{$\operatorname{Gau}_{\hat{G}}(\hat{P})$-injective} extensions of $S^G$ (defined below) is such that the existence and the obstruction problem have a complete and positive solution;
    \item there is a subclass $\operatorname{Coh}(S^G;\hat{G})$ of $\operatorname{Inj}(S^G;\hat{G})$ consisting of the so-called \textit{coherent extensions} which are $I$-coherent for maximality in the sense of Section \ref{sec_zorn} for certain $I$.
\end{enumerate}

We begin by defining what we mean by an \textit{injective extension}. 
Let $*:G\times X\rightarrow X$ be the action of a group $G$ on a
set $X$, and let $f:X\rightarrow\mathbb{R}$ be a function. We say that
$f$ is \emph{$G$-invariant in a subset $X_{0}\subset X$ }(or that
$X_{0}$ is a $G$\emph{-invariant subset for $f$})\emph{ }if $X_{0}$
is $G$-invariant and $f(g*x)=f(x)$ for every $x\in X_{0}$. This
means that $f$ factors as in the diagram below. We say that $X_{0}$ is an \emph{injective
$G$-invariant subset for $f$} if the induced map $\overline{f}_{0}$
is injective. Finally, we recall that if $X_{0}/G$ is the orbit space of a group
action, a \emph{gauge fixing} for the underlying action is a global
section for the projection map $\pi:X_{0}\rightarrow X_{0}/G$, i.e,
it is a function $\sigma:X_{0}/G\rightarrow X_{0}$ such that $\pi(\sigma(x))=x$
for every $x\in X_{0}$. In the category of sets, gauge fixings always
exists. This is a consequence of the Axiom of Choice, which states
that every surjection of sets has a section. When the group action
is implicit, we say simply that $\sigma$ is a \emph{gauge fixing
in $X_{0}$}.

$$
\xymatrix{\ar@{^(->}[d] X_0 \ar[rr]^{\pi} && X_0/G \ar@{_(->}[d] \ar@{-->}[rd]^{\overline{f}_0}\\
\ar@/_{0.5cm}/[rrr]_{f} X \ar[rr]_{\pi} && X/G & \mathbb{R} }
$$ 

Let $S^G$ be a YMT theory and let $\hat{S}^{\hat{G}}$ be an extension of $S^G$. We say that $\hat{S}^{\hat{G}}$ is \textit{injective} if its domain $\widehat{\operatorname{Conn}}(\hat{P};\hat{\mathfrak{g}})$ is an injective $\operatorname{Gau}_{\hat{G}}(\hat{P})$-invariant subset for $\hat{S}^{\hat{G}}$. The class of all of them will be denoted by $\operatorname{Inj}(S^G;\hat{G})$. The obstruction problem for this class is characterized by the following lemma, which will be proved in Subsection \ref{universal_extensions}.

\begin{lemma} \label{thm_complete_extension}
Let $S^{G}$ be a YMT theory with instanton sector $P$, let $\imath:G\hookrightarrow\hat{G}$
be a basic extension of $G$ and let $\operatorname{Conn}(\hat{P};\hat{\mathfrak{g}})\subset X_{0}\subset\Omega^{1}(\hat{P};\hat{\mathfrak{g}})$
be an injective $\operatorname{Gau}_{\hat{G}}(\hat{P})$-invariant
subset for a map $S:\Omega^{1}(\hat{P};\hat{\mathfrak{g}})\rightarrow\mathbb{R}$.
Then
\begin{enumerate}
\item Each gauge fixing $\sigma$ in $\operatorname{Conn}(\hat{P};\hat{\mathfrak{g}})$
induces a representation of the restriction $S_{0}:X_{0}\rightarrow\mathbb{R}$
of $S$ to $X_{0}$ as a complete extension of $S^{G}$ on a certain
subspace $X_{0}(\sigma)$;
\item the subspace $X_{0}(\sigma)$ is independent of $\sigma$, i.e, for
every two gauge fixings $\sigma,\sigma'$ there exists a bijection
$X_{0}(\sigma)\simeq X_{0}(\sigma')$.
\end{enumerate}
\end{lemma}

Now, let us consider the class of coherent extensions. We say $\hat{S}^{\hat{G}}$ is
 \begin{enumerate}
     \item of \emph{pullback type} if it is of the form of the last theorem, i.e, if it is complete, $\widehat{\operatorname{Conn}}(\hat{P};\hat{\mathfrak{g}})=X_0$ and there exists a gauge fixing $\sigma$ such that $$ C^1(\hat{P};\hat{\mathfrak{g}})\subset X_0(\sigma)\simeq \operatorname{Pb}(S^G,\overline{S}_{0,\jmath}).
     $$The class of all of them will be denoted by $\operatorname{Pb}(S^G;\hat{G})$;
      \item \textit{small} (or that it has has a \emph{small correction subspace}) if $C^1(\hat{P};\hat{\mathfrak{g}})$ is actually contained in $\operatorname{Conn}(\hat{P};\hat{\mathfrak{g}})$. The class of these extensions will be denoted by $\operatorname{Small}(S^G;\hat{G})$.
     \item \textit{coherent} if it is injective and of pullback-type. The class of all of them will be denoted by $\operatorname{Coh}(S^G;\hat{G})$, so that $\operatorname{Coh}(S^G;\hat{G})=\operatorname{Inj}(S^G;\hat{G})\cap \operatorname{Pb}(S^G;\hat{G})$;
     \item \textit{coherently small} if it is coherent and small. The class of them will be denoted by $\operatorname{SCoh}(S^G;\hat{G})$, so that  $\operatorname{SCoh}(S^G;\hat{G})=\operatorname{Coh}(S^G;\hat{G})\cap \operatorname{Small}(S^G;\hat{G})$
  \end{enumerate}
  
With this nomenclature, Lemma \ref{thm_complete_extension} has the following meaning:
\begin{corollary}
For every YMT $S^G$ and every basic extension $\imath :G\hookrightarrow \hat{G}$, the only obstruction to belonging to $\operatorname{Coh}(S^G;\hat{G})$ is belonging to $\operatorname{Inj}(S^G;\hat{G})$.
\end{corollary}

Notice that this corollary (and therefore Lemma \ref{thm_complete_extension}) is only about the obstruction problem because a priori we do not know anything about the existence problem for the class $\operatorname{Inj}(S^G;\hat{G})$ of injective extensions. But in Subsection \ref{sec_lemma_2} we will prove the following lemma, showing that to every extension one can assign a canonical injective extension.

\begin{lemma}\label{prop_injective}
There exists a canonical map  
$$\mathcal{I}: \operatorname{Ob}(\mathbf{Ext}(S^G;\hat{G}))\rightarrow \operatorname{Inj}(S^G;\hat{G})\cap \operatorname{Comp}(S^G;\hat{G})
$$
assigning to each extension $\hat{S}^{\hat{G}}$ of $S^G$ the  largest complete injective extension $\mathcal{I}(\hat{S}^{\hat{G}})$ for which there exists a monomorphism $\mathcal{I}(\hat{S}^{\hat{G}})\rightarrow \hat{S}^{\hat{G}}$. Furthermore, this map preserves small extensions.
\end{lemma}
Directly from the proof of Lemma \ref{thm_complete_extension} in Subsection \ref{universal_extensions} we will conclude that every pullback-type extension is small, so that $\operatorname{SCoh}(S^G;\hat{G})=\operatorname{Coh}(S^G;\hat{G})$ and we have the inclusion 
\begin{equation}\label{inclusion_lemma}
\operatorname{Pb}(S^G;\hat{G})\hookrightarrow \operatorname{SCoh}(S^G;\hat{G})=  \operatorname{Coh}(S^G;\hat{G}) . 
\end{equation}

In Subsection \ref{sec_lemma_3} we will then prove the following lemma, where we recall that a function $f:\operatorname{Ob}(\mathbf{C})\rightarrow \operatorname{Ob}(\mathbf{D})$ between objects of categories is \textit{essentially injective} if it factors through isomorphisms and the induced map $\overline{f}:\operatorname{Iso}(\mathbf{C})\rightarrow\operatorname{Iso}(\mathbf{D})$ is a bijection.

\begin{lemma}\label{lemma_retraction}
The choice of a gauge fixing $\sigma$ as in Lemma \ref{thm_complete_extension} induces a retraction $r_{\sigma}$ for the inclusion \emph{(\ref{inclusion_lemma})}. Furthermore, this retraction is an essentially injective map for the underlying full subcategories of $\mathbf{Ext}(S^G;\hat{G})$.
\end{lemma}

In the sequence, in Subsection \ref{sec_lemma_4}  we will give a proof for the following result:

\begin{lemma}\label{lemma_total_order}
The partial order $\lesssim$ is a total order when restricted to the class $\operatorname{Pb}(S^G;\hat{G})$ of pullback extensions.
\end{lemma}

As a consequence of all these steps we have the following conclusion, where $\operatorname{ED}(S^G;\hat{G})$ denotes the collection of all extended domains of a YMT $S^G$ relatively to a basic extension $\imath :G \hookrightarrow \hat{G}$.

\begin{theorem}[Theorem C]
Let $S^G$ be a YMT theory, $\imath: G\hookrightarrow \hat{G}$  a basic extension of $G$ and $I\subset \operatorname{ED}(S^G;\hat{G})$ a subset of mutually disjoint extended domains. Then every class $\mathcal{E}_0$ of extensions such that $\vert \mathcal{E}_0 \vert \geq \vert I \vert$ is internally maximal to the class $\operatorname{Coh}(S^G;\hat{G})(I)$ of coherent extensions.
\end{theorem}
\begin{proof}
Since $I$ is a subset of disjoint extended domains, its union coincides with the disjoint union and is also an extended domain. Thus, given extensions $\hat{S}^{\hat{G}}_i$, with $i\in I$, each of them with extended $i$, one can define its disjoint union $\coprod_i \hat{S}^{\hat{G}}_i$, which remains an extension of $S^G$. Furthermore, if each $\hat{S}^{\hat{G}}_i$ is injective (resp. of pullback-type), then the coproduct is injective (resp. of pullback-type too).  Therefore,  $\operatorname{Coh}(S^G;\hat{G})=\operatorname{Inj}(S^G;\hat{G})\cap \operatorname{Pb}(S^G;\hat{G})$ has $I$-coproducts. By Lemma \ref{lemma_total_order} the class $\operatorname{Pb}(S^G;\hat{G})$ is totally ordered, while by Lemma \ref{lemma_retraction} there exists an essentially injective map between pullback-type extensions and coherent extensions, showing that $\operatorname{Coh}(S^G;\hat{G})$ is totally ordered too. Thus, it is actually $I$-coherent for maximality and the result follows from Theorem \ref{thm_zorn}.
\end{proof}

\subsection{Proof of Lemma \ref{thm_complete_extension}}\label{universal_extensions}

\quad \;\,The proof consists essentially in taking the pullback between $S^{G}$
and a suitable monomorphism. Indeed, notice that we have the commutative
diagram below, where $\overline{S}_{0}$ exists and is injective due
to the injectivity hypothesis on $X_{0}$. Since the composition of
injective maps is injective, we see that $\overline{S}_{0,\jmath}=\overline{S}_{0}\circ\overline{\jmath}$
is itself injective.\begin{equation}{\label{diagram_1_thm_complete_extension}
\xymatrix{ \ar@{^(->}[d]_{\jmath} \operatorname{Conn}(\hat{P};\hat{\mathfrak{g}}) \ar[r]^-{\pi} & \operatorname{Conn}(\hat{P};\hat{\mathfrak{g}})/ \operatorname{Gau}_{\hat{G}}(\hat{P}) \ar@{_(->}[d]^{\overline{\jmath}} \\
\ar[d]_{S_0} X_0 \ar[r]_-{\pi} & X_0/ \operatorname{Gau}_{\hat{G}}(\hat{P}) \ar@/^/[ld]^{\overline{S}_{0}} \\
\mathbb{R}}} 
\end{equation}

Consider now the pullback between $S^{G}$ and $\overline{S}_{0,\jmath}$,
as in the diagram below. Due to the stability of monomorphisms under
pullbacks, it follows that $\overline{\jmath}_{0}$ is an inclusion.
Let $\sigma$ be a gauge fixing in $\operatorname{Conn}(\hat{P};\hat{\mathfrak{g}})$.
Since sections are always monomorphisms, we see that $\sigma$ is injective
and then the composition $\overline{\sigma}_{0}=\sigma\circ\overline{\jmath}_{0}$
below is an injection. Thus, the composition $\jmath\circ\overline{\sigma}_{0}$
is injective too, so that it embeds the pullback $\operatorname{Pb}(S^{G},\overline{S}_{0,\jmath})$
as a subset of $X_{0}$, which we denote by $X_{0}(\sigma)$. Due
to the commutativity of the left-hand side square, we have $S_{0}\vert_{X_{0}(\sigma)}=S^{G}\circ\delta$.
Thus, if $X_{0}(\sigma)$ is invariant under $\operatorname{Gau}_{G}(P)$,
then $S_{0}$ is a complete extension of $S^{G}$ on the correction
subspace $X_{0}(\sigma)$. This will follow from the universality
of pullbacks.
\begin{equation}{\label{diagram_2_thm_complete_extension}
\xymatrix{\ar[dd]_{\delta} \operatorname{Pb}(S^G,\overline{S}_{0,\jmath}) \ar@{-->}[r]_-{\overline{\sigma}_{0}} \ar@/^{0.7cm}/[rr]^{\overline{\jmath}_{0}} & \ar@{^(->}[d]_{\jmath} \operatorname{Conn}(\hat{P};\hat{\mathfrak{g}}) \ar[r]<0.1cm>^-{\pi} & \operatorname{Conn}(\hat{P};\hat{\mathfrak{g}})/ \operatorname{Gau}_{\hat{G}}(\hat{P}) \ar@{_(->}[d]^{\overline{\jmath}} \ar[l]<0.1cm>^-{\sigma}\\
& \ar[d]_{S_0} X_0 \ar[r]_-{\pi} & X_0/ \operatorname{Gau}_{\hat{G}}(\hat{P}) \ar@/^/[ld]^{\overline{S}_{0}} \\
\operatorname{Conn}(P;\mathfrak{g}) \ar[r]_-{S^G} & \mathbb{R}}}
\end{equation}
Notice that what we have to prove is the existence of the dotted arrow
below, where we replaced $X_{0}(\sigma)$ with $\operatorname{Pb}(S^{G},\overline{S}_{0,\jmath})$
since there is bijective correspondence between them.$$
\xymatrix{\operatorname{Gau}_G(P)\times \operatorname{Pb}(S^{G},\overline{S}_{0,\jmath}) \ar[d]_{id\times \delta} \ar@{-->}[r]^-{*}& \operatorname{Pb}(S^{G},\overline{S}_{0,\jmath}) \ar[d]^{\delta} \\
\operatorname{Gau}_G(P)\times \operatorname{Conn}(P;\mathfrak{g}) \ar[r]_-{(-)^{*}} & \operatorname{Conn}(P;\mathfrak{g})}
$$Now, let $\xi:\operatorname{Gau}_{G}(P)\rightarrow\operatorname{Gau}_{\hat{G}}(\hat{P})$ be the map discussed in Section \ref{sec_extension} and notice that the exterior diagram below
commutes, which follows precisely from the realization of the pullback
as the subset in which the functionals $S^{G}$ and $\overline{S}_{0}$
coincide. Thus, by the universality of pullbacks, there exists the
dotted arrow, which is the desired map. This concludes the first part
of the proof. $$
\xymatrix{ \ar@{-->}[rd]^{*} \ar@/^{0.8cm}/[rrrd]^{\pi \circ \xi} \ar@/_/[dddr]_{(-)^{*}\circ(id\times \delta) } \operatorname{Gau}_G(P)\times \operatorname{Pb}(S^{G},\overline{S}_{0,\jmath}) \\
& \ar[dd]_{\delta} \operatorname{Pb}(S^G,\overline{S}_{0,\jmath}) \ar[r]_-{\overline{\sigma}_{0}} \ar@/^{0.7cm}/[rr]^{\overline{\jmath}_{0}} & \ar@{^(->}[d]_{\jmath} \operatorname{Conn}(\hat{P};\hat{\mathfrak{g}}) \ar[r]<0.1cm>^-{\pi} & \operatorname{Conn}(\hat{P};\hat{\mathfrak{g}})/ \operatorname{Gau}_{\hat{G}}(\hat{P}) \ar@{_(->}[d]^{\overline{\jmath}} \ar[l]<0.1cm>^-{\sigma}\\
&& \ar[d]_{S_0} X_0 \ar[r]_-{\pi} & X_0/ \operatorname{Gau}_{\hat{G}}(\hat{P}) \ar@/^/[ld]^{\overline{S}_{0}} \\
&\operatorname{Conn}(P;\mathfrak{g}) \ar[r]_-{S^G} & \mathbb{R}} 
$$For the second part, given gauge fixing $\sigma$ and $\sigma'$,
let us prove that the induced extensions have the same correction
subspaces, i.e, there exists a bijection $X_{0}(\sigma)\simeq X_{0}(\sigma')$.
Looking at diagram (\ref{diagram_1_thm_complete_extension}) we see
that $X_{0}(\sigma)\simeq\operatorname{img}(\jmath\circ\sigma\circ\overline{\jmath})$,
while $X_{0}(\sigma')\simeq\operatorname{img}(\jmath\circ\sigma'\circ\overline{\jmath})$.
Since those maps are injective, it follows that their image are in
bijection with their domain. But both maps have the same domain: the
pullback $\operatorname{Pb}(S^{G},\overline{S}_{0,\jmath})$. Thus,
$X_{u}(\sigma)\simeq X_{0}(\sigma')$ and the proof is done.
\subsection{Proof of Lemma \ref{prop_injective}}\label{sec_lemma_2}
\quad\;\, Let $\hat{S}^{\hat{G}}$ be an extension of $S^G$ with extended domain $X\in\operatorname{ED}(S^G;\hat{G})$, correction term $C:C^1(\hat{P};\hat{\mathfrak{g}})\rightarrow \mathbb{R}$ and $\delta$-map $\delta:C^1(\hat{P};\hat{\mathfrak{g}})\rightarrow \operatorname{Conn}(P;\mathfrak{g})$. From it, define a new extension $\hat{S}_0^{\hat{G}}$ with same data, but now with a null correction term. This is obviously a complete extension of $S^G$. In turn, let $X'\subset X$ be the largest $\operatorname{Gau}_{\hat{G}}(\hat{P})$-invariant subset restricted to which $\hat{S}_0^{\hat{G}}$ remains $\operatorname{Gau}_{\hat{G}}$-invariant and such that it becomes injective on the quotient space $X'/\operatorname{Gau}_{\hat{G}}$. Since $X'$ does not necessarily contains $\operatorname{Conn}(\hat{P};\hat{\mathfrak{g}})$, it does not need to be an extended domain. On the other hand, the union $X''=X'\cup \operatorname{Conn}(\hat{P};\hat{\mathfrak{g}})$ it is. Thus, we can try to build an injective extension by considering the restriction $\hat{S}_0^{\hat{G}}\vert_{X''}$ as extended functional. Again, this may fail since a priori the correction space $C^1(\hat{P};\hat{\mathfrak{g}})\subset X$ needs not be contained in $X''$. This leads us to consider the intersection $C^1_{''}=C^1(\hat{P};\hat{\mathfrak{g}})\cap X''$ as correction subspace and the restrictions $C\vert_{C^1_{''}}$ and $\delta\vert_{C^1_{''}}$ as correction functional and $\delta$-map, respectively. It is straightforward to check that this really defines a complete and injective extension of $S^G$, which we denote by $\mathcal{I}(\hat{S}^{\hat{G}})$. By construction, if $C^1(\hat{P};\hat{\mathfrak{g}})\subset \operatorname{Conn}(\hat{P};\hat{\mathfrak{g}})$, i.e, if the starting extension is small, then the same happens to $C^1_{''}$, so that the rule $\hat{S}^{\hat{G}}\mapsto \mathcal{I}(\hat{S}^{\hat{G}})$ preserves the class of small extensions. Of course, we have a morphism $F:\mathcal{I}(\hat{S}^{\hat{G}})\rightarrow \hat{S}^{\hat{G}}$ given by the inclusions $X''\hookrightarrow X$ and $C^1_{''}\hookrightarrow C^1(\hat{P};\hat{\mathfrak{g}})$, which is clearly a monomorphism, finishing the proof.

\subsection{Proof of Lemma \ref{lemma_retraction}}\label{sec_lemma_3}
\quad\;\,Let $\hat{S}^{\hat{G}}:\widehat{\operatorname{Conn}}(\hat{P};\hat{\mathfrak{g}})\rightarrow\mathbb{R}$
be an injective extension of $S^{G}$ with small correction subspace
$C^{1}(\hat{P};\hat{\mathfrak{g}})$ and map $\delta:C^{1}(\hat{P};\hat{\mathfrak{g}})\rightarrow\operatorname{Conn}(P;\mathfrak{g})$.
We can then form the diagram below, which is commutative. For each
gauge fixing $\sigma$ we can compute the pullback and then there
exists a unique map $\mu$, as below. This map is injective. Indeed,
for every global section $\overline{\sigma}$ fo $\overline{\pi}$
one can write $S^{G}\circ\pi_{1}\circ\mu=\overline{\hat{S}^{\hat{G}}}\circ\overline{\sigma}\circ\jmath\circ\operatorname{inc}$.
Since $\hat{S}^{\hat{G}}$ is injective, the right-hand side is an
injection, which implies that the first map of the left-hand side,
$\mu$, is an injection too. By Lemma \ref{thm_complete_extension}, the pullback defines a
pullback-type extension. Thus, because $\mu$ is an injection, $\hat{S}^{\hat{G}}$
is of pullback-type. The construction above shows that the choice of a gauge fixing $\sigma$
induces a map $r_{\sigma}:\operatorname{Ext}_{0}(S^{G};\hat{G})\rightarrow\operatorname{PbExt}(S^{G};\hat{G})$.
Now, let $\hat{S}^{\hat{G}}$ be a pullback-type extension, regard
it as a injective extension and notice that the construction above
applied to it does not do anything. This means that $r_{\sigma}$
is a retraction, as desired. 
\begin{equation}{\label{diagram_retraction}
\xymatrix{\ar@{-->}[r]_{\mu} \ar@/^{0.5cm}/[rr]^{\pi \circ \operatorname{inc}} \ar@/_/[ddr]_{\delta} \ar@{^(->}@/^{1cm}/[rrr]^{\operatorname{inc}} C^{1}(\hat{P};\hat{\mathfrak{g}}) &\ar[dd]_{\pi_1} \operatorname{Pb}_{\sigma} \ar[r]^-{\pi_2}   & \ar@{^(->}[d]_{\overline{\jmath}} \operatorname{Conn}(\hat{P};\hat{\mathfrak{g}})/ \operatorname{Gau}_{\hat{G}}(\hat{P}) \ar[r]<-0.1cm>_-{\sigma} & \ar@{_(->}[d]^{\jmath} \ar[l]<-0.1cm>_-{\pi} \operatorname{Conn}(\hat{P};\hat{\mathfrak{g}})\\
&&  \widehat{\operatorname{Conn}}(\hat{P};\hat{\mathfrak{g}})/ \operatorname{Gau}_{\hat{G}}(\hat{P}) \ar[r]_-{\overline{\pi}} \ar[d]_{\overline{\hat{S}^{\hat{G}}}} & \widehat{\operatorname{Conn}}(\hat{P};\hat{\mathfrak{g}})  \ar@/^/[ld]^{\hat{S}^{\hat{G}}} \\
&\operatorname{Conn}(P;\mathfrak{g}) \ar[r]_-{S^G} & \mathbb{R}} 
}
\end{equation}

Let $\hat{S}_{1}^{\hat{G}}$ and $\hat{S}_{2}^{\hat{G}}$ be complete,
injective and with small correction subspace extensions of $S^{G}$.
Thus, for each of them we have the commutative diagram (\ref{diagram_retraction}).
If they are equivalent, both commutative diagrams are linked via
the bijections $f$ and $g$, producing a bigger diagram\footnote{This bigger diagram is best visualized
as a 3D diagram whose base and roof are the copies (\ref{diagram_retraction})
corresponding to $\hat{S}_{1}^{\hat{G}}$ and $\hat{S}_{2}^{\hat{G}}$.}. 
It is commutative due to the commutativity of (\ref{diagram_equivalence_extensions}).
Furthermore, it reduces to a diagram (\ref{diagram_equivalence_extensions})
for $\hat{S}_{i}^{\hat{G}}$ regarded as pullback-type extensions.
This means that $\hat{S}_{1}^{\hat{G}}\simeq\hat{S}_{2}^{\hat{G}}$
implies $r_{\sigma}(\hat{S}_{1}^{\hat{G}})\simeq r_{\sigma}(\hat{S}_{2}^{\hat{G}})$,
so that $r_{\sigma}$ passes to the quotient. Let $\overline{r}_{\sigma}$
be the induced map on the quotient spaces. Since $r_{\sigma}$ is
a retraction, $\overline{r}_{\sigma}$ is too. Thus, it is enough
to prove that $\overline{r}_{\sigma}$ is injective. Assume $r_{\sigma}(\hat{S}_{1}^{\hat{G}})\simeq r_{\sigma}(\hat{S}_{2}^{\hat{G}})$.
This means that we have a diagram (\ref{diagram_equivalence_extensions})
for $r_{\sigma}(\hat{S}_{i}^{\hat{G}})$. But we also have the commutative
diagram (\ref{diagram_retraction}) corresponding to each $r_{\sigma}(\hat{S}_{i}^{\hat{G}})$.
And these diagrams are linked due to the maps defining the equivalence
$r_{\sigma}(\hat{S}_{1}^{\hat{G}})\simeq r_{\sigma}(\hat{S}_{2}^{\hat{G}})$.
Therefore, we have again a bigger commutative diagram, which again
reduces to a diagram (\ref{diagram_equivalence_extensions}), now
for $\hat{S}_{1}^{\hat{G}}$ and $\hat{S}_{2}^{\hat{G}}$. Thus, $\hat{S}_{1}^{\hat{G}}\simeq\hat{S}_{2}^{\hat{G}}$,
finishing the proof.
\subsection{Proof of Lemma \ref{lemma_total_order}}\label{sec_lemma_4}
\quad\;\,We have to prove that the relation $\lesssim$ is a total order in the class of pullback-type extensions. Thus, given two pullback-type extensions we have to prove that there exists a morphism between them. This is a direct consequence of the following lemma.

\begin{lemma}
In the notations of Theorem \ref{thm_complete_extension}, let $\operatorname{Conn}(\hat{P};\hat{\mathfrak{g}})\subset X_{i}\subset\Omega^{1}(\hat{P};\hat{\mathfrak{g}})$, with $i=0,1$,
be two injective $\operatorname{Gau}_{\hat{G}}(\hat{P})$-invariant
subsets for a map $S:\Omega^{1}(\hat{P};\hat{\mathfrak{g}})\rightarrow\mathbb{R}$. If $X_0\subset X_1$, then for every gauge fixing $\sigma$ we have $X_0(\sigma) \subset X_1 (\sigma)$, where $X_i (\sigma)$ are the subsets of previous theorem. Furthermore, the diagram below commutes, meaning that $S_1$ is a complete extension of $S_0$.
\begin{equation}{\label{diagram_lemma_complete_extension}
\xymatrix{X_0(\sigma) \ar@{^(->}[rr]^-{\jmath_0 \circ \overline{\sigma}_0} \ar@{^(->}[d] && \ar@{_(->}[d] X_0 \ar[r]^{S_0} \ar[r] & \mathbb{R} \ar@{=}[d]   \\
X_1(\sigma) \ar@{^(->}[rr]_-{\jmath_1 \circ \overline{\sigma}_1} && X_1 \ar[r]_{S_1} & \mathbb{R}}}
\end{equation}
\end{lemma}

\begin{proof}[Idea of proof] The proof is a bit straightforward. The idea is the following. First build diagrams (\ref{diagram_1_thm_complete_extension}) for $X_0$ and $X_1$ and notice that they are linked by the inclusion $X_0 \hookrightarrow X_1$, as below. Since $S_0$ and $S_1$ arose as restrictions of the same functional $S$, it follows that the entire diagram commutes.
$$
\xymatrix{ \ar@/_{1cm}/[dd]_{\jmath_1} \ar@{^(->}[d]_{\jmath _0} \operatorname{Conn}(\hat{P};\hat{\mathfrak{g}}) \ar[r]^-{\pi} & \operatorname{Conn}(\hat{P};\hat{\mathfrak{g}})/ \operatorname{Gau}_{\hat{G}}(\hat{P}) \ar@{_(->}[d]^{\overline{\jmath _0}} \ar@/^{1.5cm}/[dd]^{\overline{\jmath _1}} \\
\ar@/_{1cm}/[dd]_{S_0} \ar@{^(->}[d] X_0 \ar[r]^-{\pi} & X_0/ \operatorname{Gau}_{\hat{G}}(\hat{P}) \ar@{_(->}[d] \ar@/_{0.5cm}/[ddl]_{\overline{S}_{0}}\\
\ar[d]_{S_1} X_1 \ar[r]_(.5){\pi} & X_1/ \operatorname{Gau}_{\hat{G}}(\hat{P}) \ar@/^/[ld]^{\overline{S}_{1}} \\
\mathbb{R}}
$$
Taking the pullback of each diagram as in (\ref{diagram_2_thm_complete_extension}) and using the universality of pullbacks we get an universal map $\eta:\operatorname{Pb}(S^{G},\overline{S}_{0,\jmath _0}) \rightarrow \operatorname{Pb}(S^{G},\overline{S}_{1,\jmath_1}) $ such that $\overline{\jmath}_0 = \overline{\jmath}_1 \circ \eta $. Since $\overline{\jmath}_0$ is a monomorphism, this implies that $\eta$ is a monomorphism too and therefore an injection. The commutativity of the entire diagram involving both pullbacks implies the commutativity of (\ref{diagram_lemma_complete_extension}), finishing the proof. 
\end{proof}

\section*{Acknowledgements} Y. X. Martins was supported by CAPES (grant number 88887.187703/2018-00) and L. F. A. Campos was supported by CAPES (grant number 88887.337738/2019-00).

\bibliographystyle{abbrv}
\bibliography{references}

\end{document}